\documentclass[11pt]{amsart}

\usepackage{amsmath,amsfonts}
\usepackage[table]{xcolor}
\usepackage{graphicx,amssymb,verbatim,epsf,amsmath,setspace,float,dsfont, soul}
\usepackage[numbers, comma, round, sort&compress]{natbib}
\newcommand{\bi}{\begin{itemize}}
\newcommand{\ei}{\end{itemize}}
\usepackage{tabularx}

\newtheorem{theorem}{Theorem}[section]

\usepackage[font=footnotesize,labelfont=bf]{caption}
\usepackage{anysize}
\marginsize{2.3cm}{2.3cm}{2cm}{3.3cm}

\usepackage{multicol}
\makeatletter

\newcommand{\bC}{\mathbb{C}}

\newcommand{\bR}{\mathbb{R}}

\newcommand{\sP}{\mathcal{P}}

\title[Decomposing the parameter space of biological networks]{Decomposing the parameter space of biological networks\\ via a numerical discriminant approach}

\author{
Heather A. Harrington, Dhagash Mehta,\\
Helen M. Byrne, Jonathan D. Hauenstein
}
\address{Mathematical Institute, The University of Oxford, Oxford, UK \newline\indent
Department of Applied and Computational Mathematics and Statistics, University of Notre Dame, Notre Dame, IN, USA.}
\begin{document}


\begin{abstract}
Many systems in biology, physics and engineering can be described by systems of ordinary differential equation containing many parameters. When studying the dynamic behavior of these large, nonlinear systems, it is useful to identify and characterize the steady-state solutions as the model parameters vary, a technically challenging problem in a high-dimensional parameter landscape. Rather than simply determining the number and stability of steady-states at distinct points in parameter space, we decompose the parameter space into finitely many regions, the steady-state solutions being consistent within each distinct region.  From a computational algebraic viewpoint, the boundary of these regions 
is contained in the discriminant locus.  
We develop global and local numerical algorithms for 
constructing the discriminant locus and  classifying the parameter landscape. 
We showcase our numerical approaches by applying them to  molecular and cell-network models.\\
keywords: parameter landscape, numerical algebraic geometry, discriminant locus,  dynamical systems, cellular networks
 \end{abstract}

\maketitle
\section{Introduction}
The dynamic behavior of many biophysical systems can be mathematically modeled 
with systems of differential equations that describe 
how the state variables interact and evolve over time.  The differential equations 
typically include parameters that represent physical processes such as kinetic rate constants, the strength of cell-cell interactions, and external stimuli, and the behavior of the state variables
may change as the parameters vary.  Typically, determining and classifying all steady state
solutions of such nonlinear systems, as a function of the parameters, is a difficult problem. However, when the equations
are polynomial, or can be translated into polynomials, which is the case for many biological, physics and engineering systems,
computing the steady state solutions 
becomes a problem in computational algebraic geometry.  With this setup, it is possible to compute the regions of the parameter space that give rise to different numbers of steady-state solutions. Here we present an approach 
based on the numerical algebraic geometric paradigm 
of computing points rather than defining equations and inequalities, surpassing current symbolic approaches, with a theoretical guarantee of finding all steady-states.

Due to the ubiquity of such problems, many methods have been proposed
for identifying and characterizing steady state solutions 
over a parameter space.  Some examples include using 
a cylindrical algebraic decomposition \cite{Collins} with
related variants \cite{LR07,xia2007discoverer,brown2003qepcad} and
computing the ideal of the discriminant locus using resultants or Gr\"obner basis methods, 
(see \cite{CLO05,GKZ,Sturmfels}).
Unfortunately, each of these methods
has drawbacks that prevent them from being applicable to models 
for real-world applications due to their algorithmic 
complexity, symbolic expression swell, and inherent sequential
nature.  On the other hand, numerical continuation methods as implemented in, for example,
AUTO \cite{doedel1981auto} and MATCONT \cite{dhooge2003matcont}, 
compute bifurcations for a parametric model and
are local in that one ``continues'' from a given initial point.
Recently, it has become possible to compute all solutions 
over the complex numbers $\bC$ to a system of polynomial equations 
using homotopy continuation and, 
more generally, numerical algebraic geometry,
(see \cite{bates2013numerically,SWbook,wampler1990numerical}).
%
%
Such methods have been implemented in software packages
including Bertini \cite{bates2008software}, 
HOM4PS-3 \cite{HOM4PS3}, and PHCpack \cite{verschelde1999algorithm}
with Paramtopy~\cite{Paramotopy} extending Bertini to 
study the solutions at many points in parameter space.
Typically, these methods work over $\bC$ while the solutions
of interest in biological models are in a subset of the real numbers $\bR$,
e.g., one is interested in steady-states in the positive orthant where the variables are~biologically~meaningful.

We present a {\em numerical discriminant locus} method for decomposing parameter space into distinct solution regions. Recall that when solving the roots of the equation $ax^2+bx+c=0$, where $a$, $b$, and $c$ are parameters, and $x$ is the variable, 
the discriminant locus defined by $\Delta := b^2-4ac = 0$ is the 
boundary separating regions in which the two distinct solutions for $x$ are 
real ($\Delta>0$) and nonreal ($\Delta<0$).  We propose three methods
for decomposing the parameter space.  The first, aimed at one-dimensional parameter
spaces, is a sweeping approach that will locate the finitely-many
real points in the discriminant locus building upon~\cite{3Dtumor,PV10}
thereby yielding finitely many regions, which are open intervals in this case, 
where the number of steady-state solutions
is consistent.  The second, for low-dimensional parameter spaces, 
provides a complete decomposition of the parameter space into finitely many regions
after decomposing the discriminant locus.  
Since computing and decomposing the discriminant locus may be impractical
for high-dimensional parameter spaces, our third method 
uses the sweeping approach to compute a local decomposition of the 
parameter space near a given point in the parameter space.  
When decomposing a high-dimensional parameter space is desirable, one could bootstrap together the local analyses to generate a more complete, or global, view of the parameter space.

In the next section, we introduce the discriminant locus and present
the proposed algorithms. Then we apply our algorithms to two biological models.

\section{Discriminant Locus and the Algorithms}
\label{sec:disc_and_alg}

We consider autonomous systems of differential equations of the form
\begin{equation} 
\frac{d}{dt} \boldsymbol{x} = \boldsymbol{f}(\boldsymbol{x},\boldsymbol{p}) \label{eq:ode} 
\end{equation}
where $ \boldsymbol{x} = (x_1,\dots,x_N)$ 
is the collection of state variables,
$\boldsymbol{p} = (p_1,\dots,p_s)$ is the collection of system parameters,
and $ \boldsymbol{f}(\boldsymbol{x},\boldsymbol{p})$
is a system of $N$ polynomial or rational functions.
For $\boldsymbol{p}\in\bR^s$, the steady state solutions to~Eq.~\ref{eq:ode} 
are defined as $ \boldsymbol{x}\in\bR^N$ such that $\boldsymbol{f}(\boldsymbol{x},\boldsymbol{p}) = 0$.  
We are particularly interested in the case when 
$\boldsymbol{f}(\boldsymbol{x},\boldsymbol{p}) = 0$ 
has finitely many solutions in $\bC^N$, all of which are nonsingular, 
for generic~$\boldsymbol{p}$ as this is typically the case for biological networks.

The parameter space $\sP\subset\bR^s$ for Eq.~\ref{eq:ode} 
consists of the values of the parameters $\boldsymbol{p}$ which one
aims to consider, here those that are biologically meaningful.
For simplicity, one should consider $\sP = \bR^s$ 
or the positive orthant in $\bR^s$.  
The quantitative behavior of the steady state solutions 
is constant on regions inside the parameter space $\sP$,
e.g., the number of physically realistic steady state solutions is the same for all parameter values in the region.
One can also refine the quantitative behavior, by considering, for example, the number of positive steady state solutions that are locally stable.
The boundaries of these regions are contained in the {\em discriminant~locus}.

Classically, the discriminant locus consists of the parameter
values for which nongeneric behavior occurs.  
Suppose that $\boldsymbol{p}\in\sP$ is such that $\boldsymbol{f}(\boldsymbol{x},\boldsymbol{p}) = 0$ has the generic behavior,
that is, $\boldsymbol{f}(\boldsymbol{x},\boldsymbol{p}) = 0$ has the expected number of nonsingular isolated solutions.
By the implicit function theorem, this generic behavior extends to an
open neighborhood containing~$\boldsymbol{p}$.  One can keep increasing the size
of this neighborhood in the parameter space until it touches the discriminant locus.
When focused on counting the number of steady state solutions, 
the discriminant locus is defined by
the closure of all values of $\boldsymbol{p}\in\sP$ such that there 
exists~$\boldsymbol{x}$ with $\boldsymbol{f}(\boldsymbol{x},\boldsymbol{p}) = 0$
and $J_{\boldsymbol{x}}\boldsymbol{f}(\boldsymbol{x},\boldsymbol{p})$ is not invertible, where $J_{\boldsymbol{x}}\boldsymbol{f}$ is the Jacobian matrix of $\boldsymbol{f}$
with respect to the state variables.

One can add to the classical discriminant locus other values of the parameters
based on the quantitative property of interest,
which we collectively call the {\em discriminant locus} for the problem.
For example, the number of positive steady state solutions
occurs at either the classical discriminant locus or the closure
of the set of $\boldsymbol{p}\in\sP$ such that there exists $\boldsymbol{x}$ with 
$\boldsymbol{f}(\boldsymbol{x},\boldsymbol{p}) = 0$ and $x_i = 0$
for some $i\in\{1,\dots,N\}$.  If a function $f_i$ is rational, 
one also has to consider the closure of the set of $\boldsymbol{p}\in\sP$ where
the numerators of $\boldsymbol{f}(\boldsymbol{x},\boldsymbol{p})$ are zero and the denominator of $f_i$ is zero.  
Since the discriminant locus is contained in a hypersurface in $\sP$,
the parameter set $\sP$ with the discriminant locus removed consists of regions 
for which the quantitative behavior of the solution set is constant.

\subsection{Algorithm and Implementation}

Our approach builds on methods in numerical algebraic geometry, e.g.,
see \cite{bates2013numerically,SWbook,wampler1990numerical} for a general overview.
In particular, rather than aiming to find equations that vanish
on the discriminant locus and inequalities that describe 
the regions in parameter space, we search for points lying on the discriminant
locus and describe regions based on giving a point lying in the region.  
That is, since each region is connected, one can start with one point 
in a region and trace out the boundaries, which are contained in the discriminant locus, for that region.  In this way, rather than finding explicit
equations and inequalities that define the boundaries of a region,
one starts with a point in the region and has a method that traces
out the boundaries as needed.  This paradigm is similar to the numerical ``cell decomposition'' approach for real 
curves \cite{RealCurves} and surfaces \cite{RealSurfaces1,RealSurfaces2},
rather than a symbolic approach where the boundaries
are described by polynomials \cite{Collins,LR07} (see \cite{niu2008algebraic,hanan2010stability,hernandez2011towards,GHRS15} for applications to biology).

The first approach, which we call {\em perturbed sweeping},
is for one-dimensional parameter spaces, i.e., $s = 1$.  
In this case, the discriminant locus for the problem of interest
consists of at most finitely many points.  Sweeping refers to tracking
the solution curves $\boldsymbol{x}(p)$ where $\boldsymbol{f}(\boldsymbol{x}(p),p) = 0$ as $p\in\bR$ varies via
continuation.  To avoid numerical issues with attempting to track through the
discriminant locus, we propose to sweep along a perturbed path.  
For $i = \sqrt{-1}$ and $\epsilon\in\bR$, we consider the perturbed solution curves
$\boldsymbol{x}_\epsilon(p)$ defined by $\boldsymbol{f}(\boldsymbol{x}_\epsilon(p),p+\epsilon i) = 0$.

\begin{theorem}\label{thm:perturbed}
With the setup described above, for all but finitely many $\epsilon\in\bR$, 
all perturbed solution curves are smooth.
\end{theorem}
\begin{proof}
Since there are only finitely many points in the discriminant locus, 
there can be only finitely many values of $\epsilon\in\bR$ such that there 
exists $\delta\in\bR$ with $\delta + \epsilon i$ in the discriminant locus.
\end{proof}

Since $\boldsymbol{x}_\epsilon(p)\rightarrow \boldsymbol{x}(p)$ as $\epsilon\rightarrow0$, we are able
to recover information about the actual solution curves with the distinct
numerical advantage of tracking {\em smooth} solution curves. 
Importantly, our sweeping approach evades possible numerical issues from \cite{PV10} 
which arise as the number of state variables increase 
and are associated with monitoring the determinant of the Jacobian matrix with respect to the state variables.  That is, the determinant of
the Jacobian matrix can be ``close'' to zero for matrices which are 
``far'' from being rank deficient.  For example, 
consider the matrix $A_n = 2^{-1}\cdot I_n$,
where $I_n$ is the $n\times n$ identity matrix.
For any $n\geq1$, $A_n$ is half a unit away from the nearest singular 
matrix while $\det(A_n) = 2^{-n}$.  To avoid this situation,
one monitors the condition number as in \cite{3Dtumor}; moreover 
we utilize a perturbation to regularize the sweeping path.
Thus, rather than track directly along the real line  
and possibly pass through a parameter point lying exactly on the 
discriminant locus where the Jacobian matrix is rank deficient,
we track along a nearby path where the Jacobian matrix has full rank.  
This avoids possible issues associated with tracking 
where the Jacobian matrix is rank deficient, but still permits
the location of the singularities by monitoring the condition number.
If further refinement is needed, additional efficient
local computations can be employed, e.g., \cite{GS85,GH15}.

We now build on this perturbed sweeping approach to generate methods
for parameter spaces which are not one-dimensional.  
Our next approach, applicable for low-dimensional parameter spaces, 
is a global approach called a {\em global region decomposition}
that mixes projections, critical sets,
and perturbed sweeping.  We start in the two-dimensional case, $s = 2$,
for which the discriminant locus is contained in a curve.  
For a (sufficiently) random projection $\pi(\boldsymbol{p}) = \alpha_1 p_1 + \alpha_2 p_2$, 
consider intersecting the discriminant locus with the family of lines
defined by $\pi(\boldsymbol{p}) = t$.  Using the perturbed sweeping approach 
on the added parameter~$t$ along the discriminant locus
yields the critical points of the discriminant locus with respect to $\pi$.
The discriminant locus has the same quantitative behavior between
its critical points and thus the regions can be constructed by slicing
between the critical points and at each of the critical points, following
a modification of \cite{RealCurves}.

For a global region decomposition for higher dimensional parameter spaces
one reduces the number of dimension by considering the critical point set. 
For example, in the three-dimensional case, $s = 3$, 
and $\pi(\boldsymbol{p})$ consisting to two (sufficiently) random projections,
one first considers the so-called critical curve of the discriminant locus
with respect to $\pi$, namely the set of points on the discriminant locus
where the Jacobian is singular with respect to $\pi$.  
One then decomposes the curve case as above yielding a region decomposition
in $\pi(\sP)$ which one then lifts to a region decomposition for $\sP$.

Since a global region decomposition is not practical for high-dimensional 
parameter spaces, we propose a local region decomposition method 
by combining perturbed sweeping with the classical approach 
of ray tracing.  Given a point $\boldsymbol{p}\in\sP$ not contained in the discriminant locus,
one can view the codimension one components of the discriminant locus by
using the perturbed sweeping approach along emanating from $\boldsymbol{p}$.  
Once points on the discriminant locus are found, one can use homotopy
continuation methods to track along the boundaries and locate other regions.
This method is local in the sense that small regions could be missed
but avoids the expense of computing critical points of projections
in the global approach.  By bootstrapping local decompositions from 
various $\boldsymbol{p}$, one can aim towards generating a 
global view of the parameter space if desired.

\begin{figure}[htb!]
\includegraphics{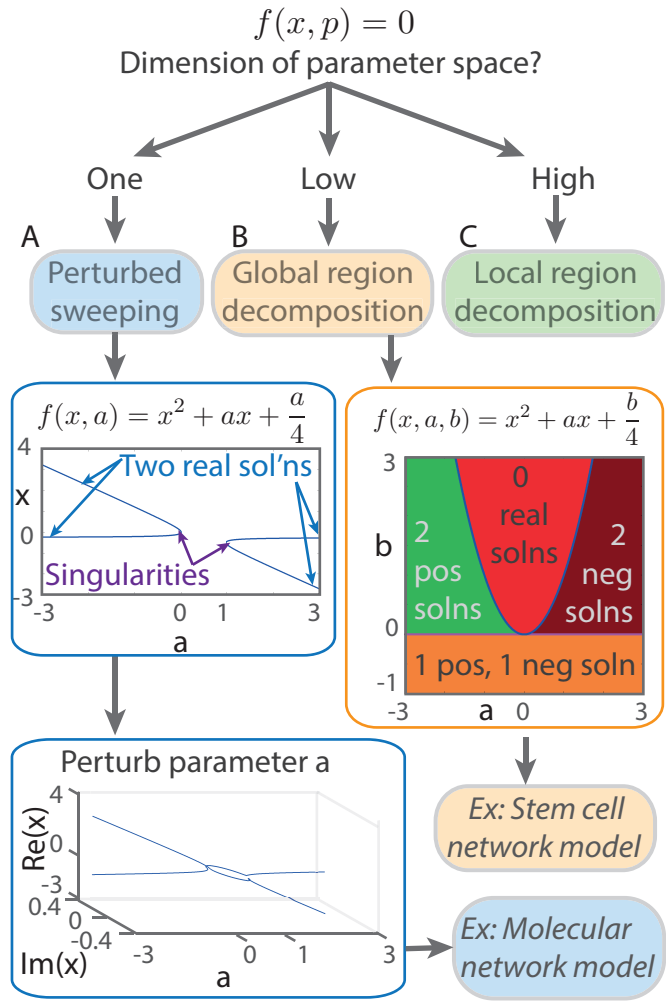}
\caption{
Flow chart of methods. {\bf(A)} Linear perturbation. {\bf(B)} Global method for decomposing $(a,b)$ parameter space into regions of number of real steady-states. {\bf(C)} Local method for high dimensional parameter space analysis.
}
\label{fig-1}
\end{figure}

\subsection{Quadratic Example}

To illustrate the perturbed sweeping and global region decomposition approach, 
we consider two examples of a parameterized quadratic equation.
The first has one parameter, namely $f(x,a) = x^2 + ax + a/4$.  
The classical discriminant for quadratic polynomials
yields $\Delta = a^2 - a$ with discriminant locus
$a = 0,1$.  In particular, $f = 0$ has 
a singular real solution when $a = 0$ or $a = 1$, 
two distinct real solutions when $a < 0$ or $a > 1$, 
and two distinct nonreal solutions when $0 < a < 1$.  
To avoid tracking through the singularities, 
we use the perturbed sweeping method 
using $f(x,a+\epsilon i) = 0$.  
In this case, for any fixed nonzero $\epsilon\in\bR$ and
$a\in\bR$, $f(x,a+\epsilon i) = 0$ always has 2 distinct solutions.
By taking $\epsilon$ near zero, say $\epsilon = 10^{-6}$, 
we sweep along the smooth curve parameterized by $a$ and 
observe the expected solution behavior as shown in Figure~\ref{fig-1}.

The second example has two parameters, namely $f(x,a,b) = x^2 + ax + b/4$
and we aim to decompose the parameter space based on the number of real
and positive solutions, which is typical in biological problems.  
The classical discriminant for quadratic polynomials
yields $\Delta = a^2 - b$ with the sign condition adding to this
via the equation $f(0,a,b) = b/4 = 0$, i.e., $b = 0$.  
Thus, the discriminant locus in this case consists of two irreducible curves
which cut the parameter space $(a,b)\in\bR^2$ into four regions 
where the number of real, positive, and negative
solutions are constant on these regions as shown in Figure~\ref{fig-1}.  

By using the projection $\pi(a,b) = a$, the perturbed sweeping method
finds the critical point $a = 0$ of the discriminant locus.
For any $a < 0$ or $a > 0$, there are three regions in $b$, namely
$b < 0$, $0 < b < a^2$, and $b > a^2$.  For $a = 0$, there are two regions
in $b$, namely $b < 0$ and $b > 0$ which must extend away from $a = 0$ by the implicit
function theorem.  Starting from one point in each of these regions, continuation 
is used to merge these $7$ regions into $4$ distinct regions of $\bR^2$.

\section{Results}\label{sec:results}

We showcase our methods by applying them to two illustrative biological models: a detailed ODE model of the gene and protein signaling network that induces long-term memory proposed by Pettigrew {\em et al.}~\cite{Pettigrew:2005ba}, and a network model of cell fate specification in a population of interacting stem cells. Since cellular decision making often depends on the number of accessible (stable) steady-states that a system exhibits, we seek to identify distinct regions of parameter space that can elicit different system behavior.

\subsection{Molecular network model}

A gene and protein network for long term memory was proposed by Pettigrew et al \cite{Pettigrew:2005ba} and investigated using bifurcation and singularity analysis by Song et al. \cite{Song:2006ku}. 
The model developed is of the form of Eq.~\ref{eq:ode} where $\boldsymbol{f}$ 
consists of a total of $15$ polynomial and rational functions,
$\boldsymbol{x} \in \bR^{15}$ is the vector of 15 model variables, 
and $\boldsymbol{p} \in \bR^{40}$ is a vector of 40 parameters. 

\begin{figure}[h!]
\includegraphics{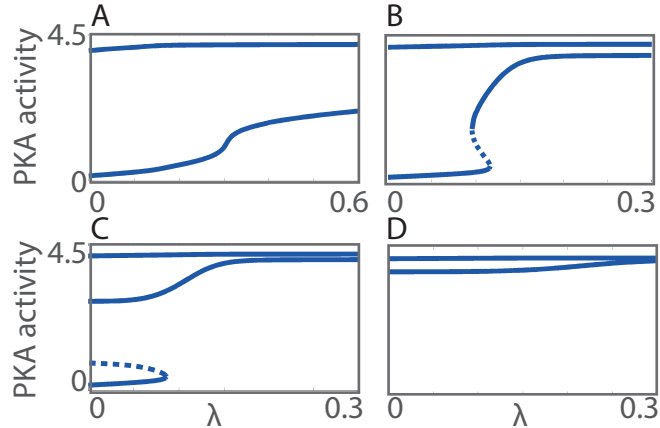}
\caption{Perturbation sweeping method applied to molecular network of long term memory. $\lambda$ is the bifurcation parameter. Solid lines denote a stable steady-state, dashed lines denote an unstable steady-state. {\bf(A)} Parameter $k_{34} = 0.0022.$ {\bf(B)} Parameter $k_{34} = 0.015.$ {\bf(C)} Parameter $k_{34} = 0.03.$ {\bf(D)} Parameter $k_{34} = 0.1.$
}
\label{fig-bifs}
\end{figure}

The full system $\boldsymbol{f}$ is given in Appendix~\ref{appendix:Song}
for which the equations $\boldsymbol{f}(\boldsymbol{x},\boldsymbol{p})=0$
has $432$ isolated nonsingular solutions for general $\boldsymbol{p}$.
The variable of interest required for long term facilitation is the steady-state of protein kinase A (PKA) in response to the extracellular stimulus parameter $\lambda$. Our aim is to demonstrate the sweeping perturbation method on a large model and reproduce the results in Fig.~5 of \cite{Song:2006ku}. We verify the two stable states and one unstable state, reproducing this figure. However we also find another solution not reported, demonstrating the power of this method (see top branch of Fig.~\ref{fig-bifs}). On inspection, this additional steady state is not on the same branch and is not biologically feasible so we can reject it as nonphysical. From this exhaustive first step of identifying and characterizing all the steady-states, one must exercise caution and systematically check each solution. 
%

\subsection{Cellular network model}
Most multicellular organisms emerge from a small number of stem-like cells which become increasingly specialized as they proliferate until they transition to one of a finite number of differentiated states \cite{clevers}.  
We propose a caricature model of cell fate specification for a ring of cells, and investigate how cell-cell interactions, mediated by diffusive exchange of a key growth factor, may affect the number of (stable) configurations or patterns that the differentiated cells may adopt.  The model serves as a good test case for discriminant locus methods since, by construction, there is an upper bound on the number of feasible steady states ($2^N$ solutions for a ring of $N$ cells) and some of these patterns are equivalent due to symmetries inherent in the governing equations. 
%

In more detail, we consider a ring of interacting cells ($i=1, \dots, N$) and denote by $x_i(t) \geq 0$ the concentration within cell $i$ of a growth factor or protein (e.g., notch), whose value
determines that cell's differentiation status \cite{fre,sprinzak,yeung,Visvader:2016fu,Grun:2015ev}. The subcellular dynamics of $x_i$ are represented by a phenomenological function $q(x_i)= - (x-\varepsilon) (a + x) (1 - x)$ with $0 < \varepsilon < a < 1$, this function guaranteeing bistability of each cell in the absence of cell-cell communication. The bistability represents two distinct cell fates; e.g., high and low levels of notch may be associated with differentiation of intestinal epithelial cells into secretory and absorptive phenotypes differentiation \cite{fre,sprinzak,clevers}.
We assume further that cell $i$ communicates with its nearest neighbors (cells $i\pm1$) via diffusive exchange of $x_i$ and denote by parameter $g \geq 0$ the coupling strength. Thus, our cell network model can be written 

\begin{equation}
\frac{d x_i}{d t} = q(x_i) + g \cdot \sum_{j= i-1}^{i+1} (x_j - x_i), 
\text{ for all $i=1, \dots, N$.}
\label{eq:cell-network}
\end{equation}
\begin{figure}[h!]
\includegraphics[width=\textwidth]{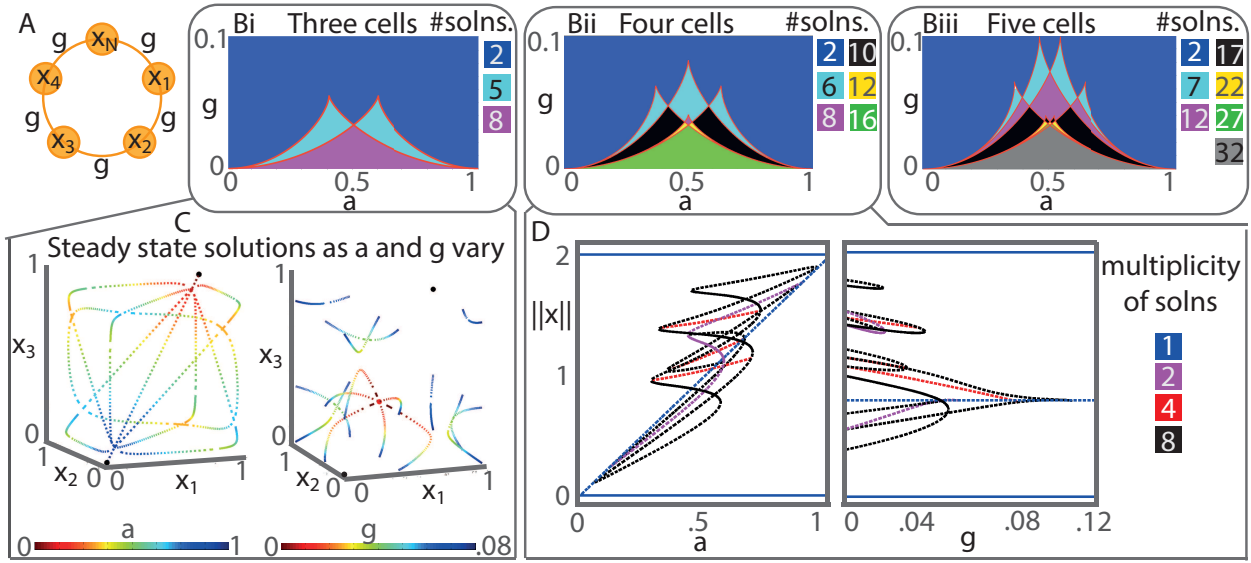}
\caption{
Global region decomposition applied to coupled cell network model. {\bf(A)} Ring of cell network. Each cell $x_i$ has bistable dynamics at stable states at $0$ and $1$, and unstable state given by $a$. The cells are coupled to neighbor cells by a coupling strength parameter~$g$. {\bf(B)} Region decomposition with parameters $a,g$, similar to a classification diagram, with the number of real stable steady-stages denoted by different colors. The number of stable real steady-states is given for each network where $N = 3,4,5$. {\bf(C)} Steady-state values in state space for $N = 3$ cell network are plotted as parameter $a$ is varied between $[\varepsilon,1]$ or $g$ is varied between 0 and 0.08. The two black dots are always a stable steady-states, the red is always an unstable steady-state, these steady-states are independent of parameter values $a$ and $g$. As the parameter value varies, the color changes according to the color bar.  {\bf(D)} Bifurcation diagram showing the 2-norm of $x$ as parameters are varied  
}
\label{fig-cells}
\end{figure}
We impose periodic boundary conditions so that $x_{N+1} \equiv x_1$ and $x_0 \equiv x_N$, as shown in Fig.~3A.
We analyze the model by applying the global decomposition method for $N=3,4,5$ cells, and construct classification diagrams in $(a,g)$ parameter space (Fig.~3B). We notice that intermediate values of $a$ generate the largest number of real stable steady-state; small and high values of $a$ yield fewer real stable steady-states. Interestingly, all cells synchronize for intermediate to strong values of the coupling parameter $g$ (two stable states in blue region in Fig~\ref{fig-cells}B). We conclude that strong cell-cell communication reduces the number of stable steady state configurations that a population of cells can adopt and, thus, cell-cell communication could be used robustly to drive the cells to a small number of specific states. When coupling is weak ($0 <g<0.1$), the interacting cells have more flexibility in terms of their final states, with the 5-cell network admitting up to 32 stable steady-states (Fig~\ref{fig-cells}B(iii)).  Weaker cell-cell communication allows more patterns to emerge and may be appropriate when it is less important that neighboring cells share the same phenotype. We find that the regions of $(a,g)$ parameter space that give rise to more than two (synchronized) steady states also increase in size as the number of cells increases. 

In addition to decomposing parameter space into regions based on the multiplicity of steady solutions, the method also provides valuable information about how solution structure and stability change as system parameters vary. For example, in Fig.~3C for the $N=3$ cell network, we show how the values and stability of the steady-states for $(x_1,x_2,x_3)$ change as $a$ varies with $g = 0.025$ and as $g$ varies with 
\mbox{$a = 0.4$}.  In Fig.~3D, we plot bifurcation diagrams as $a$ and $g$ vary 
as before for the $N=4$ cell network; instead of presenting particular components $x_i$ ($i=1,2,3,4$), we plot the 2-norm ($\|\boldsymbol{x}\| = (x_1^2 + x_2^2 + x_3^2 + x_4^2)^{1/2}$) to capture the multiplicity of solutions. We note that for $N=3$ and $N=4$ there are always two stable and one unstable steady-states, independent of $(a,g)$ parameter values (as shown by the black and red points in Fig.~3C, and by the solid blue lines in Fig.~3D). We demonstrate the local region decomposition method to analyze a generalization of Eq.~\ref{eq:cell-network} to be cell specific by considering different cell-cell couplings (see Appendix~\ref{appendix:cell}). Through this caricature stem-cell model, we have classified the steady-states, determined which are parameter-independent, explored changes in steady-state behavior as parameters are varied through the distinct regions of the parameter landscape. 




\section{Conclusion}

We have presented a suite of new methods, based on computational algebraic geometry, for decomposing the parameter space associated with a dynamical system into distinct regions based on the multiplicity and stability of its steady-state solutions. The methods enable us to understand the parameter landscape of high-dimensional, ordinary differential models with large numbers of parameters. 
These methods have considerable potential: they could be used to analyze differential equation models associated with a wide range of real-world problems in biology, science and engineering which cannot be tackled with~existing~approaches.

We have demonstrated that groups of cells, especially cells that can have bistable internal dynamics, when coupled, increase the number of real stable steady-states. First we have considered all the cells to be homogeneous, and even with this simple model, we can gain new insight into how stem cells may transit into fully differentiated cells, with much richer dynamics, just by considering the interaction of the cells as a network. We considered different couplings between cells using a local region decomposition in the Appendix.  We remark that it would be interesting to understand how a mutation in a single cell (for example, set $a$ very small) compares to real biological systems. In the future, it would be 
beneficial to consider different networks of 
cells (e.g., two growth factors, generalizations to 2D lattices) 
and the dynamical regimes these systems can admit. 

%

By embracing and exploiting the complexity naturally present in 
biology, which often translates to nonlinearity in biological models, there is an incredible opportunity to 
develop and apply numerical algebraic geometric methods to study 
such models.  With nonlinearity arising in all areas of science and engineering,
the increasing complexity of various biological models provides a fertile 
testing ground for nonlinear approaches with a long-term goal
of making numerical algebraic geometric methods  
as ubiquitous as methods in numerical linear algebra.

\section*{Acknowledgement}
We thank J. Byrne and G. Moroz for for helpful discussions.
JDH was supported in part by NSF ACI 1460032, Sloan Research Fellowship, and Army Young Investigator Program (YIP). HAH acknowledges funding from AMS Simons Travel Grant and EPSRC Fellowship EP/K041096/1.

\bibliographystyle{pnas}
\bibliography{bibliofile}
\appendix

\section{The Song et al. model}\label{appendix:Song}

For completeness, we reproduce below the system of 15 ordinary differential equations proposed by Song et al. \cite{Song:2006ku}. The parameters we vary are $\lambda=[5-HT]$ and $k_{34}$ = $k_{\mbox{ApSyn}}$. All other parameters are fixed to values given in \cite{Song:2006ku}.

{\scriptsize
\begin{eqnarray}
\frac{d[cAMP]}{dt} &=& 3.6 \left(\frac{[5-HT]}{[5-HT]+14}\right) - ([cAMP] - 0.06)\\
\frac{d[R]}{dt} &=& V_{\mbox{syn}} + k_{\mbox{fpka}}[RC] [cAPM]^2 - k_{\mbox{bpka}}[R][C]\nonumber \\
& & - k_{\mbox{Ap-Uch}}[R]([\mbox{Ap-Uch}]_{\tau} - \mbox{Ap-Uch}_{\mbox{basal}}) - k_{\mbox{dpka}}[R] \\
\frac{d [C]}{dt} &=& V_{\mbox{syn}} + k_{\mbox{fpka}}[RC] [cAPM]^2 - k_{\mbox{bpka}}[R][C]\nonumber \\
& &+ k_{\mbox{Ap-Uch}}[RC]([\mbox{Ap-Uch}]_{\tau} - \mbox{Ap-Uch}_{\mbox{basal}}) - k_{\mbox{dpka}}[C]\\
\frac{d [RC]}{d t} &=& k_{\mbox{bpka}} [R][C] - k_{fpka}[RC] [cAMP]^2 \nonumber \\
& & - k_{\mbox{Ap-Uch}}[RC]([\mbox{Ap-Uch}]_{\tau} - \mbox{Ap-Uch}_{\mbox{basal}}) - k_{\mbox{dpka}}[RC]\\
\frac{d [REG]}{d t} &=& k_{\mbox{translation}}[\mbox{mRNA}_{\mbox{REG}}] ERK_{\mbox{act}} \nonumber \\ 
& & - v_{\mbox{dreg}}\left(\frac{[REG]}{[REG] + K_{\mbox{dreg}}} \right) - k_{\mbox{dsm}}[REG] \\ 
\frac{d [pREG]}{dt} &=& v_{\mbox{rphos}}ERK_{\mbox{act}} \left(\frac{[REG] - [pREG]}{[REG]-[pREG] + K_{\mbox{rphos}}} \right) - v_{\mbox{dreg}} \left(\frac{[pREG]}{[REG] + K_{\mbox{dreg}}} \right) \nonumber \\
& & - k_{\mbox{dsm}}[\mbox{pREG}] \\
\frac{d [\mbox{mRNA}_{\mbox{REG}}]}{d t} &=& v_{\mbox{mREG}} - v_{\mbox{dmreg}}[5-HT] \left(\frac{[\mbox{mRNA}_{\mbox{REG}}]}{[\mbox{mRNA}_{\mbox{REG}}] + K_{\mbox{dmreg}}}\right) - k_{\mbox{dmreg}} [\mbox{mRNA}_{\mbox{REG}}] \\
\frac{d [\mbox{Raf}]}{d t} &=& - k_{f, \mbox{Raf}} [5-HT] [Raf] + k_{b,\mbox{Raf}} [\mbox{Raf}^{p}] \\
\frac{d [\mbox{MAPKK}]}{d t} &=& -k_{f, \mbox{MAPKK}}[\mbox{Raf}^{p}] \left(\frac{[\mbox{MAPKK}]}{[\mbox{MAPKK}] + K_{\mbox{MK}}} \right) + k_{b,\mbox{MAPKK}}\left(\frac{[\mbox{MAPKK}^{p}]}{[\mbox{MAPKK}^{p}] + K_{\mbox{MK}}} \right)\\
\frac{d[MAPKK]^{pp}}{d t} &=& k_{f,\mbox{MAPKK}}[\mbox{Raf}^{p}]\left(\frac{[\mbox{MAPKK}^{p}]}{[\mbox{MAPKK}^{p}] + K_{\mbox{MK}}} \right)
- k_{b,\mbox{MAPKK}}\left(\frac{[\mbox{MAPKK}^{pp}]}{[\mbox{MAPKK}^{pp}] + K_{\mbox{MK}}} \right) \\
\frac{d [\mbox{ERK}]}{d t} &=& - k_{f,\mbox{ERK}} [\mbox{MAPKK}^{pp}] \left( \frac{[\mbox{ERK}]}{[\mbox{ERK}] + K_{\mbox{MK}}} \right)
+ k_{b, \mbox{ERK}} \left( \frac{[\mbox{ERK}^{p}]}{[\mbox{ERK}^{p}] + K_{\mbox{MK}}}\right)\\
\frac{d [\mbox{ERK}^{pp}]}{d t} &=&  k_{f,\mbox{ERK}} [\mbox{MAPKK}^{pp}] \left( \frac{[\mbox{ERK}^{p}]}{[\mbox{ERK}^{p}] + K_{\mbox{MK}}} \right)
- k_{b, \mbox{ERK}} \left( \frac{[\mbox{ERK}^{pp}]}{[\mbox{ERK}^{pp}] + K_{\mbox{MK}}}\right) \\
\frac{d P_{\mbox{pka}}}{d t} &=& k_{\mbox{phos1}} \mbox{PKA}_{\mbox{act}} (1 - P_{\mbox{pka}}) - k_{\mbox{dephos1}} \mbox{PPhos} P_{\mbox{pka}}\\
\frac{d P_{\mbox{erk}}}{d t} &=& k_{\mbox{phos2}} \mbox{ERK}_{\mbox{act}} (1 - P_{\mbox{erk}}) - k_{\mbox{dephos2}} \mbox{PPhos} P_{\mbox{erk}}\\
\frac{d [\mbox{AP-Uch}]}{d t} &=& k_{\mbox{ApSyn}}\left(\frac{P^{2}_{\mbox{pka}}}{P^{2}_{\mbox{pka}} + K^{2}_{\mbox{pka}}} \right)
\left(\frac{P^{2}_{\mbox{erk}}}{P^{2}_{\mbox{erk}} + K^{2}_{\mbox{erk}}} \right) + k_{\mbox{ApSynBasal}} - k_{\mbox{deg}}[\mbox{Ap-Uch}]
\end{eqnarray}
}
where
{\scriptsize
\begin{eqnarray*}
 PKA_{\mbox{act}} &=& [C], \\ 
 \left[\mbox{Raf}^{p} \right] &=& [ \mbox{Raf}_{\mbox{tot}} ] - [ \mbox{Raf} ], \\
 k_{\mbox{bpka}}  &=& \left(\frac{k_{\mbox{baspka}}}{1+[pREG]/K_{\mbox{reg}}} \right), \\ 
 \left[\mbox{ERK}^{p}\right] &=& [\mbox{ERK}_{\mbox{tot}}] - [\mbox{ERK}] - [\mbox{ERK}^{pp}],  \\
 \mbox{ERK}_{\mbox{act}} &=& [\mbox{ERK}^{pp}] + \mbox{ERK}_{\mbox{basal}}, \\ 
\left[\mbox{MAPKK}^{p}\right] &=& [\mbox{MAPKK}_{\mbox{tot}}] - [\mbox{MAPKK}] - [\mbox{MAPKK}^{pp}]. 
\end{eqnarray*}
}
Following \cite{Song:2006ku}, we fix the model parameters at the following values:
{\scriptsize
\begin{eqnarray*}
 V_{\mbox{syn}} &=& 0.002 \mbox{min}^{-1},
 k_{\mbox{fpka}} = 105.0 \mbox{min}^{-1},\nonumber \\
 k_{\mbox{Ap-Uch}} &=& 0.007 \mbox{min}^{-1}, 
 \mbox{Ap-Uch}_{\mbox{basal}} = 0.10, \nonumber \\
 k_{\mbox{dpka}} &=& 0.00048 \mbox{min}^{-1}, 
 \tau = 250 \mbox{min}, \nonumber \\
 k_{\mbox{baspka}} &=& 12 \mbox{min}^{-1}, 
 k_{\mbox{reg}} = 0.00064,\nonumber \\
 k_{\mbox{translation}} &=& 4 \mbox{min}^{-1},
 v_{\mbox{rphos}} = 1 \mbox{min}^{-1}, K_{\mbox{rphos}} = 1.5, \nonumber \\
 v_{\mbox{dreg}} &=& 0.16 \mbox{min}^{-1}, 
 K_{\mbox{dreg}} = 0.0015, \nonumber \\
 k_{\mbox{dsm}} &=& 0.02 \mbox{min}^{-1}, 
 v_{\mbox{mREG}} = 2\time 10^{-5} \mbox{min}^{-1}, \nonumber \\
 v_{\mbox{dmreg}} &=& 0.00225 \mbox{min}^{-1}, 
 K_{\mbox{dmreg}} = 0.01,\nonumber \\
 k_{\mbox{dmreg}} &=& 3 \times 10^{-5} \mbox{min}^{-1}, 
 k_{f, \mbox{Raf}} = 0.00345 \mbox{min}^{-1}, \nonumber \\
 k_{f, \mbox{MAPKK}} &=& 0.7 \mbox{min}^{-1}, 
 k_{f, \mbox{ERK}} = 0.44 \mbox{min}^{-1}, \nonumber \\
 k_{b, \mbox{Raf}} &=& 0.001 \mbox{min}^{-1}, 
 k_{b, \mbox{MAPKK}} = 0.12 \mbox{min}^{-1}, \nonumber \\
 k_{b, \mbox{ERK}} &=& 0.12 \mbox{min}^{-1},
  [\mbox{Raf}_{\mbox{tot}}] = 0.5, \nonumber \\
  \mbox{MAPKK}_{\mbox{tot}} &=& 0.5, 
[ERK_{\mbox{tot}}] = 0.5 \nonumber \\
K_{\mbox{MK}} &=& 0.08, 
\mbox{ERK}_{\mbox{basal}} = 0.015, \nonumber \\
k_{\mbox{phos1}} &=& 0.01 \mbox{min}^{-1}, 
k_{\mbox{dephos1}} = 1.5 \mbox{min}^{-1}, \nonumber \\
\mbox{PPhos} &=& 0.1, 
k_{\mbox{phos2}} = 0.005 \mbox{min}^{-1},\nonumber \\
k_{\mbox{ApSyn}} &=& 0.02 \mbox{min}^{-1}, 
k_{\mbox{ApSynBasal}} = 0.0009 \mbox{min}^{-1}, \nonumber \\
K_{\mbox{pka}} &=& 0.2,
K_{\mbox{erk}} = 0.004, \nonumber \\
k_{\mbox{deg}} &=& 0.01 \mbox{min}^{-1}
\end{eqnarray*}
}
We remark that the discriminant method can handle rational functions.
For this example, the denominators do not vanish near the regions of 
interest so they do not have any impact on the behavior of the solutions.
If the denominator also vanished when finding a solution to the system of equations from the numerators, then the parameter values for which this occurs would be added into the discriminant~locus.

\section{Cell network model}\label{appendix:cell}

The model in Eq.~\ref{eq:cell-network} used uniform coupling $g$ for
all nearest neighbors.  Using \cite{HS10}, 
as a set, the classical discriminant locus for this model 
with $N = 4$ has degree $72$.
That is, there is a degree $72$ polynomial $\Delta(a,g)$ such
that the classical discriminant locus is defined by $\Delta = 0$.
With $a = 0.4$ as in Fig~3D, the univariate polynomial equation
$\Delta(0.4,g) = 0$ has $45$ complex solutions, $25$ of which are real with $15$ positive with only~$4$ corresponding to a change in the 
number of stable steady-state solutions yielding the 
regions (intervals) for $g$ approximately:
$${\small
\begin{array}{r|ccccc} 
\mbox{regions} & [0,0.0197) & (0.0197,0.0206) & (0.0206,0.0411) & (0.0411,0.0533) & (0.0533,\infty) \\
\hline 
\mbox{\# stable steady-state solns} & 16 & 12 & 10 & 6 & 2 
\end{array}
}$$

We now consider a generalization of this model
which uses coupling strengths $g_{i,i+1}=g_{i+1,i}\geq0$ 
between cell $i$ and $i+1$ with cyclic 
ordering ($N+1\equiv 1$), namely
\begin{equation}
\frac{d x_i}{d t} = q(x_i) + \sum_{j= i-1}^{i+1} g_{i,j} \cdot (x_j - x_i), 
\text{ for all $i=1, \dots, N$.}
\label{eq:cell-network2}
\end{equation}
With $N = 4$, as a set, the classical discriminant locus for this generalized 
model has degree $486$.  We consider the case with $a = 0.4$ and 
$g_{4,1} = g_{1,4} = 0$ leaving three free 
parameters $g_{1,2}$, $g_{2,3}$, and $g_{3,4}$.  

The perturbed sweeping approach along the ray defined
by $g_{i,i+1} = i\cdot t$ for $i = 1,2,3$ and $t\geq0$
decomposes the space $t\geq0$ into $15$ regions, approximately
$$
\begin{array}{c|c} 
\mbox{regions} & \mbox{\# stable steady-state solns} \\
\hline
[0,0.0079) & 16 \\
(0.0079,0.0080)& 15 \\
(0.0080,0.0132)& 14 \\
(0.0132,0.0136)& 13 \\
(0.0136,0.0137) & 12 \\
(0.0137,0.0142)& 11 \\
(0.0142,0.0153)& 10 \\
(0.0153,0.0161)& 9 \\
(0.0161,0.0171)& 8 \\
(0.0171,0.0237)& 7 \\
(0.0237,0.0352& 6 \\
(0.0352,0.0360)& 5 \\
(0.0360,0.0407)& 4 \\
(0.0407,0.1264)& 3 \\
(0.1264,\infty)& 2 
\end{array}
$$
As a comparison, the classical discriminant with respect to $t\in\bC$ 
consists of $312$ distinct points of which $84$ are real with $42$ positive
with $14$ corresponding to a change in the number stable steady-state solutions.
Figure~\ref{fig-sweep} plots the region with $16$ 
stable steady-state solutions along various rays emanating from the origin.

\begin{figure}[h!]
\includegraphics[scale=0.4]{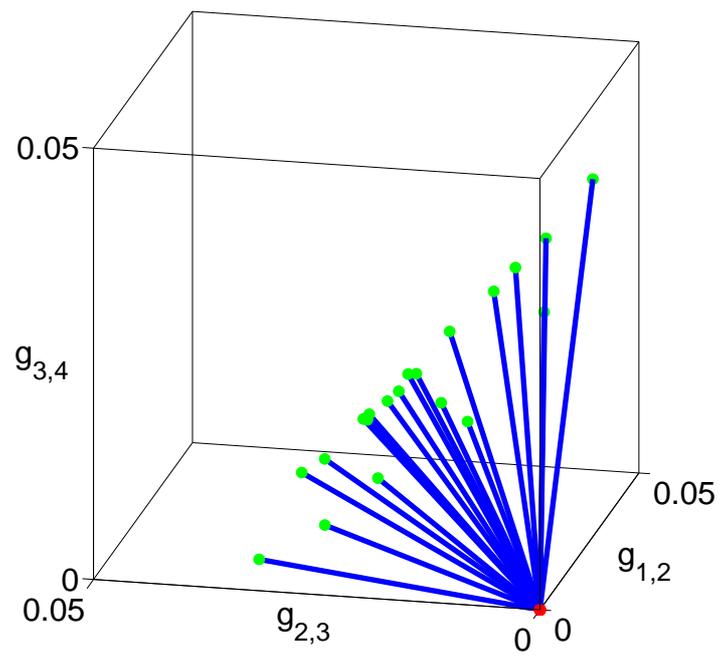}
\caption{Plot of the regions with $16$ stable steady-state solutions
along various rays emanating from $g_{1,2} = g_{2,3} = g_{3,4} = 0$.}
\label{fig-sweep}
\end{figure}

\end{document}